%% file: checkX10-2.tex
\documentclass[9pt,preprint]{sigplanconf}

\usepackage{listings}
\usepackage{caption}
\usepackage{subcaption}
\usepackage{sidecap}
\usepackage{amsmath}
\usepackage{threeparttable}

\usepackage{flushend}

\usepackage{url}
\usepackage{pgf}

\newenvironment{proof}{\begin{quotation} {\bf Proof}}
{\noindent\rule{0.5em}{2.0ex}\end{quotation}}

\makeatletter
\let\@copyrightspace\relax
\makeatother

\def\Loc{\texttt{Loc}}
\def\Val{\texttt{Val}}
\def\withmath#1{\relax\ifmmode#1\else{$#1$}\fi}
\newcommand{\derives}{\longrightarrow}
\newcommand{\starderives}{\stackrel{\star}{\longrightarrow}}
\def\tuple#1{\withmath{\langle #1 \rangle}}
\def\alt{\withmath{\;\char`\|\;}}

\def\?{\withmath{\;\char`\?\;}}

\newcommand\LL[1]{\withmath{\llb#1\rrb}}

\newtheorem{definition}{Definition}[section]
\newtheorem{proposition}{Proposition}[section]

\newcommand{\defeq}{\stackrel{def}{=}}
\newcommand\Skip{\texttt{skip}}

\def\llb{\withmath{\lbrack\!\lbrack}}
\def\rrb{\withmath{\rbrack\!\rbrack}}

\newcommand\isasync[1]{\texttt{isasync}\ #1}
\newcommand\issync[1]{\texttt{issync}\ #1}
\newcommand\async[1]{\texttt{async}\ #1}

\newcommand\stuck[1]{\texttt{stuck}\ #1}

\newcommand\finish[1]{\texttt{finish}\ #1}

\newcommand\Advance[0]{\texttt{advance}}

\newcommand\FOR[4]{\texttt{for}(#1\ \texttt{in}\ #2..#3)\ #4}
\def\O#1{{\mathcal O}\LL{#1}}

\def\P#1{{\mathcal P}\LL{#1}}

\newcommand\restr[2]{{
  \left.\kern-\nulldelimiterspace 
  #1 
  \vphantom{\big|} 
  \right|_{#2} 
  }}

\def\from#1\infer#2{{{\textstyle #1}\over{\textstyle #2}}}
\def\rname#1\from#2\infer#3{{{\textstyle #2}\over{\textstyle #3}}{\ \mbox{#1}}}
\def\rnames#1\side#2\from#3\infer#4{{{\textstyle #3}\over{\textstyle #4}}{\ \textstyle(#2)}{\ \textstyle(#1)}}
\def\axname#1\axiom#2{{\textstyle #2}{\ \mbox{#1}}}
\def\code#1{\texttt{#1}}

\def\hbk{\prec\!\!\prec}

\title{Checking Race Freedom of Clocked X10 Programs}

\authorinfo{Tomofumi Yuki}{INRIA Rennes}
{tomofumi.yuki@inria.fr}
\authorinfo{Paul Feautrier}{LIP (ENS Lyon, INRIA, CNRS, UCBL)}
           {paul.feautrier@ens-lyon.fr}
\authorinfo{Sanjay Rajopadhye}{Colorado State University}
           {Sanjay.Rajopadhye@colostate.edu}
\authorinfo{Vijay Saraswat}{IBM Research}{vijay@saraswat.org}

\begin{document}

\maketitle

\input{abstract}

\section{Introduction}
\input{introduction}

\section{The X10 Subset}\label{sec:new-clocks}
\input{new-clocks}

\section{Array Dataflow Analysis for Polyhedral X10 Programs}\label{sec:background}
\input{background}

\section{Semantics of Clocks}\label{sec:clocks}
  \input{semantics}
\subsection{The Clocked Happens-Before Relation}\label{sec:happens-before}
\input{clocked-H-B}

\section{Clocks and Races}\label{sec:races}
\input{races}

\section{Examples}\label{sec:examples}
\input{examples}

\section{Implementation}\label{sec:implementation}
\input{implementation}

\section{Related Work}\label{sec:related}
\input{related}

\section{Conclusions and Future Work}\label{sec:conclusion}
\input{conclusion}

\bibliographystyle{plain}
{\footnotesize
\bibliography{checkX10-2}}
\end{document}

%% file: abstract.tex
\begin{abstract}

One of many approaches to better take advantage of parallelism, which has now
become mainstream, is the introduction of parallel programming languages.
However, parallelism is by nature non-deterministic, and not all parallel bugs
can be avoided by language design. This paper proposes a method for
guaranteeing absence of data races in the polyhedral subset of clocked X10
programs.

Clocks in X10 are similar to barriers, but are more dynamic; the subset of
processes that participate in the synchronization can dynamically change at
runtime. We construct the happens-before relation for clocked X10 programs,
and show that the problem of race detection is undecidable. However, in many
practical cases, modern tools are able to find solutions or disprove their
existence. We present a set of benchmarks for which the analysis is possible
and has an acceptable running time.
\end{abstract}
\vspace{-0.25cm}

%% file: introduction.tex

Driven by the limitations of current micro-architecture technologies,
parallel computing has gone mainstream. Parallel programs are now required
to utilize the massive amount of parallelism provided by multi-core and many-core architectures.
Efficiently parallelizing programs by hand requires significant effort, and the programmers must ``think parallel''.
However, automatic parallelization is extremely difficult, and has seen limited success so far.

As an alternative approach, many parallel programming languages are being
developed~\cite{X10, yelick1998titanium, chamberlain2007parallel,
cave2011habanero}.  These languages aim for both high \emph{performance} and
high \emph{productivity} by employing new programming models.  Parallel
programming is inherently more difficult than sequential programming,
especially when it comes to debugging.  This is due to the non-deterministic
nature of parallelism that brings a new set of bugs that are not consistently
reproducible. 

There are a number of parallel debuggers (e.g.,~\cite{krammer2004marmot,
park2011efficient, vetter2000umpire}) to help detecting parallel bugs, but
their use is time consuming, and give no guarantees.  We seek to complement
dynamic debuggers by statically analyzing parallel programs and providing
guarantees.

Specifically, we target the X10 parallel programming language~\cite{X10}.
X10 provides (logical) deadlock-free guarantee for programs written following a
small set of rules. However, data race is another class of common parallel bugs
that remains to be detected. Static analyses of X10 programs for data race
detection have been proposed~\cite{agarwal2007may, vechev2010automatic,
yuki2013array}, but are currently limited to programs without \emph{clocks}.

Clock is a synchronization mechanism in X10 that is similar to barriers.
However, it is more \emph{dynamic}, meaning that the processes participating in
a synchronization can dynamically change. For example, consider the following:

\begin{tabular}{c|c}
\begin{lstlisting}
//Process 1
   sync;
   S0;
   sync;
\end{lstlisting}
&
\begin{lstlisting}
//Process 2
   sync;
   S1;
	
\end{lstlisting}
\end{tabular}

The program above depicts an abstracted view of two parallel processes
synchronizing through calls to {\tt sync} statements.  If the {\tt sync}
statement corresponds to MPI-style barrier, there is a deadlock, because the
two processes do not call the barrier the same number of times.  In X10, Process~2
will be removed from the set of processes participating in the synchronization
when it reaches its termination, and therefore the above does not cause a 
deadlock.  This dynamic behavior of X10 clocks invalidates static analyses
developed for MPI-style barriers~\cite{aiken1998barrier, kamil2005concurrency,
zhang2007barrier}. 

In this paper, we present static race detection for X10 programs that can
provide race-free guarantees for regions of programs amenable to our analysis.
Race-free guarantee of program sub-regions, and the ability to detect parallel
bugs at compile-time both contribute to reduce the debugging effort, and hence
leads to increased productivity. The main limitation of our approach is that
it is \emph{polyhedral}; loop bounds and array accesses must be affine. This
class of programs can be represented using the \emph{polyhedral model}, a
mathematical framework for reasoning about program
transformations~\cite{Feau:2011}. Although the applicability of the model is
limited, it has been proven to be effective in automatic parallelization.

Specifically, we extend the work by Yuki et al~\cite{yuki2013array} by
extending the subset to all of its core parallel constructs; {\tt async}, {\tt
finish}, and \emph{clocks}. We retain the main strength of their work that the
analysis is statement \emph{instance-wise}, and array \emph{element-wise}.
Instance-wise analysis provides information at the granularity of statement
instances: the execution of a statement for a specific value of 
the loop iterators. Element-wise analysis distinguishes array accesses at the granularity
of array elements, so that accesses are flagged as in conflict only when the
same element is accessed.

Our key contributions are:
\begin{itemize}
 \item Extension of the operational semantics to include clocks.
       The ``happens-before'' relation due to clocks is defined. Intuitively, 
       the happens-before relation requires to count how many times a process has
       synchronized in its execution trace.
 \item Formulation of the happens-before relation with clocks in a 
       polyhedral context. The number of times a process has synchronized 
       at a given statement instance can be formulated as counting integer 
       points in a union of polyhedra.
 \item Proof of undecidability of race detection for clocked X10 programs. In
       the general case, the synchronization counts produce \emph{polynomials}. 
       Thus, comparing two statement instances with the extended 
       happens-before relation turns out to be undecidable.
 \item Race-free guarantee of clocked X10 programs by disproving all possible 
       races. Since the happens-before relation involves polynomials in 
       the general case, Integer Linear Programming commonly used in 
       the polyhedral context cannot be used. Instead, we formulate the 
       race condition as a constraint statisfaction problem and resort to advanced 
       SMT solvers to \emph{disprove} the existence of potential races.
 \item Prototype implementation of the above.
\end{itemize}  

The rest of the paper is organized as follows. We introduce the subset of X10
we use in this paper in Section~\ref{sec:new-clocks}.  
We review the work of Yuki
et al.~\cite{yuki2013array} that we extend upon in
Section~\ref{sec:background}.  We develop the operational semantics and define
the happens-before relation of clocked X10 programs in
Section~\ref{sec:clocks}.  We establish the connection between the semantics
and polyhedral analysis and present the data race detection formulation in
Section~\ref{sec:races}.  We demonstrate our approach with examples in
Section~\ref{sec:examples} and present our implementation in
Section~\ref{sec:implementation}.  Related work is discussed in
Section~\ref{sec:related} and we conclude in Section~\ref{sec:conclusion}.

%% file: new-clocks.tex
The subset of X10~\cite{X10} considered in a related work~\cite{yuki2013array}
that deals with X10 programs without clocks is the following:
\begin{itemize} 
\item Sequence ($\{S\,T\}$): Composes two statements in sequence.
\item Sequential \texttt{for} loop: All loops have an associated
  loop iterator. X10 loops may scan a multidimensional iteration
  space, but such loops are expanded.
\item \texttt{async}: The body of the \texttt{async} is executed by an
independent lightweight thread, called \texttt{activity} in X10.
\item \texttt{finish}: The activity executing a {\tt finish} waits for all
  activities spawned within the body of \texttt{finish} to terminate before
  proceeding further.
\end{itemize}
The details of the leaf statements are not relevant; the only required 
information is the array (elements) that are read and written.

Additionally, the program must fit the polyhedral model.  The polyhedral model
requires loop bounds, and array access to be affine expressions of the
surrounding loop indices and program parameters (run-time constants).

\subsection{X10 Clocks}
Clocks are generalization of barriers for dynamically varying sets of
activities (X10 threads).  Activities can be \emph{registered}
to clocks, and may synchronize by calls to {\tt advance} statements that blocks
until all activities registered to the clock executes an {\tt advance}.  X10
allows for the set of registered activities to dynamically change.

One subtle but important difference with global barriers is that if all but one
activity had reached an {\tt advance}, and the last activity reaches its
termination, the activity is \emph{de-registered} from its clocks, allowing the
remaining activities to proceed.

\subsection{Explicit and Implicit Clocks}
X10 has two different syntaxes for clocks; explicit and implicit clocks.
In the explicit syntax, clocks are represented as objects in the program.
The programmer has fine grain control over when to synchronize on a 
particular clock, and can even de-register an activity from some clocks.

However, the use of the explicit syntax is not necessary in the majority 
of the cases. 
For instance, we were able to rewrite a significant set of benchmarks
(see Section~\ref{sec:benchmarks}) using only implicit clocks.
Explicit syntax can lead to deadlocks, although it can easily be avoided 
by following simple rules~\cite{X10}, and needs a difficult points-to analysis
for its verification.
Thus, programmers are strongly encouraged to use the alternative syntax:
 implicit clocks.

\subsection{Implicit Clock Syntax}
In implicit clock syntax both {\tt finish} and {\tt async} may be annoted 
with an optional keyword {\tt clocked}.
The {\tt clocked} variants must satisfy the following:
\begin{itemize}\itemsep0pt
\item {\tt clocked async} must be enclosed by {\tt clocked finish} or 
{\tt clocked async}
\item {\tt clocked async} must not be enclosed by 
\emph{unclocked} {\tt async} or \emph{unclocked} {\tt finish}
\item {\tt advance} must be enclosed by {\tt clocked finish}
\item {\tt advance} must not be enclosed by \emph{unclocked} {\tt async}
\end{itemize}
and has the following semantics:
\begin{itemize}\itemsep0pt
\item {\tt clocked finish} creates a clock and registers itself to it,
\item {\tt clocked async} registers itself to the clock created by its 
governing {\tt clocked finish},
\item {\tt clocked finish} de-registers itself from the clock it created 
when reaching the end of its enclosed statement,
\item {\tt clocked async} de-registers itself from its clock when reaching the 
end of its enclosed statement, and
\item {\tt advance} synchronizes with the clock created by its 
governing {\tt clocked finish}.
\end{itemize}
where the governing {\tt clocked finish} is the first {\tt clocked finish} found
when traversing the AST towards the root from a node.

\subsection{Advance Counts}
Informally, the clock may be considered to have as many associated counters as
there are registered activities.  When all participating activities execute an
{\tt advance}, synchronization takes place, and all counters are incremented.
The statements that may-happen-in-parallel are restricted to those that are
executed when the value of the counters match. This notion is formalized in
Section~\ref{sec:happens-before}.

\subsection{Example}
Figure~\ref{fig:clock-ex} illustrates simple uses of X10 clocks using implicit
syntax. Note the dynamic changes in the activities participating to the clock.
In Figure~\ref{fig:clock-ex1-code}, all activities blocks until the primary
activity reaches the end of {\tt finish}. The number of participants gradually
decreases, since they do not have the same trip count.
Figure~\ref{fig:clock-ex2-code} is an example of the opposite behavior, where
the number of participants gradually increases.


\begin{figure}[h]
        \begin{subfigure}[b]{0.4\columnwidth}
\begin{lstlisting}
clocked finish
   for (i=1:N)
      clocked async {
         for (j=i:N)
            advance;
            S0;
      }

\end{lstlisting}
		\caption{\label{fig:clock-ex1-code}Barrier-like synchronization with implicit clocks}
        \end{subfigure}
\hspace{0.1\columnwidth}
        \begin{subfigure}[b]{0.4\columnwidth}
                \centering
                \includegraphics[width=1\columnwidth]{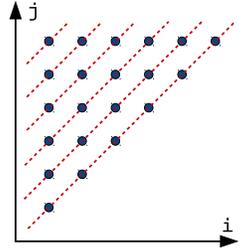}
                \vspace{-0.3cm}
		\caption{\label{fig:clock-ex1-is}Parallel iterations of \ref{fig:clock-ex1-code}}
        \end{subfigure}

        \begin{subfigure}[b]{0.4\columnwidth}
\begin{lstlisting}
clocked finish
   for (i=1:N)
      clocked async {
         for (j=i:N)
            advance;
            S0;
      }
      advance;
\end{lstlisting}
		\caption{\label{fig:clock-ex2-code}Slightly different synchronization pattern}
        \end{subfigure}
\hspace{0.1\columnwidth}
        \begin{subfigure}[b]{0.4\columnwidth}
                \centering
                \includegraphics[width=1\columnwidth]{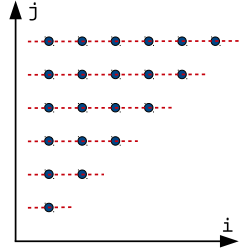}
                \vspace{-0.3cm}
		\caption{\label{fig:clock-ex2-is} Parallel iterations of \ref{fig:clock-ex2-code}}
        \end{subfigure}
\caption{\label{fig:clock-ex}Examples of X10 clock usage.
Figure~\ref{fig:clock-ex1-code} illustrates barrier-like synchronization where
$N$ activities spawned by the {\tt i} loop synchronizes every step of the {\tt
j} loop. Every spawned activity immediately reaches an {\tt advance}, and
blocks until the primary activity reaches the end of the {\tt finish} block,
de-registering itself from the clock. Figure~\ref{fig:clock-ex2-code} is a
slight variation, where the primary activity execute an {\tt advance} after
each spawn.  Thus, the spawned activities may proceed before all $N$ activities
are spawned, leading to a different parallel execution.}
\end{figure}

%% file: background.tex

Our work is an extension to the work by Yuki et al.~\cite{yuki2013array}, which
only handles {\tt async} and {\tt finish}.  These authors present a concise (and
affine) expression of the happens-before relation for X10 programs with {\tt
async} and {\tt finish}. The happens-before relation is then used to extend
array dataflow analysis~\cite{Feau:91} that finds precise (statement
\emph{instance}-wise and array \emph{element}-wise) dependence information.  We
briefly describe the key ideas that we reuse in this paper.

\subsection{Paths and Iteration Vectors} \label{sect:paths}
The statement \emph{instances} are identified by a vector of integers, 
called an iteration vector. This vector is computed (symbolically) as a 
path from the root to nodes in the AST.
As the AST is traversed, values are appended to the vector based on the 
following rules:
\begin{itemize}
\item {\tt Sequence}: Integer $x$ when taking the $x$-th branch of a sequence.
\item {\tt For}: Loop iterator $i$.
\item {\tt Async}: $a$
\item {\tt Finish}: $f$
\end{itemize}
The iteration vector for a statement instance is obtained by instantiating 
the loop iterators to integer values.
This is similar to the conventional iteration vectors used in the 
polyhedral literature, with the addition of $a$ and $f$. Due to structural
constraints, when two iteration vectors are compared, $a$ and $f$ are never 
compared to anything but themselves, and thus their order is irrelevant.

\subsection{Happens-Before Relation}
The happens-before relation for X10 programs with {\tt finish} and {\tt async}
is formulated as an incomplete lexicographic order.  For sequential programs,
the full lexicographical order denotes the execution order.  In the presence of
parallelism, the execution order is no longer total. Yuki et
al.~\cite{yuki2013array} show that for the {\tt finish/async} subset of X10
programs, the happens-before relationship can be expressed as an incomplete
lexicographic order.

The strict lexicographic order of two distinct such iteration $u$ and $v$ 
is defined as follows:
\begin{eqnarray}\small
  u \ll v & \equiv & \bigvee_{p \ge 0} u \ll_p v \label{depthwise},\\
  u \ll_p v & \equiv & \left( \bigwedge_{k=1}^{p} u_k = v_k \right) 
                    \wedge (u_{p+1} < v_{p+1})
\end{eqnarray}

%

The incomplete version restricts the dimensions to be compared to some subset 
$I$. Intuitively, the set $I$ is constructed such that the dimensions that do not
contributed to happens-before relation, due to concurrent execution, are removed.
The happens-before relation (denoted as $\prec$) is defined as follows:
\begin{equation} \label{incomplete:lex}
 u \prec v \equiv \bigcup_{k \in I} u \ll_k v
\end{equation}

\subsection{Race Detection through Dataflow Analysis}\label{sec:ppopp-ada}
Using the iteration vectors and happens-before ordering, the authors develop
an extension to the array dataflow analysis~\cite{Feau:91} for X10 programs.
Array dataflow analysis finds the statement instance that produced the 
value used by an instance of a read.

Given reader and writer statements $R,W$, and memory access functions $f_R,f_W$,
the set of potential sources is defined by:
\begin{eqnarray}
   r & \in & D_R \label{context}, \\
   w & \in & D_W, \label{existence}\\ 
   f_W(w) & = & f_R(r), \label{conflict}\\
\neg (r \prec w) & \wedge & r \not = w \label{order}
\end{eqnarray}
where 
\begin{itemize}
\item Constraints \ref{context} and \ref{existence} restrict the statement 
instances to its domain (the set of legal iterations),
\item Constraint \ref{conflict} restricts to those that access the same 
array \emph{element}, and 
\item Constraint \ref{order} excludes writers that happen after reads, 
and writes by the same statement instance.
\end{itemize}
In a sequential program, $\prec$ is total, hence 
$\neg (v \prec w) \wedge v \not = w \equiv w \prec v$, which is the 
usual formulation \cite{Feau:91}.

The above gives a set of writer instances $w$ that may be a producer for a 
read instance $r$ for a single writer statement $W$.
The proposed analysis proceeds by finding the most recent 
$w$ among all statements that write to the same array.
Since the happens-before relation is not total, the most recent 
$w$ may not be unique, and there is a race when a producer cannot be uniquely 
identified.
There are two kinds of races:
\begin{itemize} 
\item Read-Write Race: When an instance of reader ($r$) and an instance of
writer ($w$) are not ordered. The write may or may not happen before the read,
and thus the result is not deterministic.
\item Write-Write Race: When two writer instances, not necessarily of the same
statement, is not ordered among themselves, but both of them happen-before the reader instance.
\end{itemize}

Our proposed approach consists of disproving all races found
by using additional happens-before order deduced from clocks.

%% file: semantics.tex
In this section, we define the operational semantics, and the happens-before
relation for the subset of X10 we consider.

\subsection{Operational Semantics}\label{sec:sem}
We provide a simple, concise structural operational semantics (SOS) for the
fragment of X10 considered in this paper. The semantics is based on
the same ideas as  \cite{yuki2013array}, but extends it in two
ways. It provides a treatment of full sequential composition
(permitting, e.g.{} \code{\{\{s1; s2;\}s3;\}} -- 
\cite{yuki2013array} permits only \code{\{s1;\{s2; s3;\}}). More
importantly, it provides a formal treatment for clocks that is
considerably simpler than \cite{x10-concur05}.

We assume that a set of (typed) locations \Loc, and a set of values, \Val, is
given.  \Loc{} typically includes the set of variables in the program under
consideration.  With every d-dimensional $N_1\times\ldots \times N_d$
array-valued variable {\tt a} of type array are associated a set of distinct
locations, designated {\tt a(0,\ldots,0), \ldots, a($\mathtt{N_1}$-1, \ldots,
  $\mathtt{N_d}$-1)}.
The set of values includes integers and arrays.

A {\em heap} is a partial (finite) mapping from {\Loc} to \Val. For $h$ a
heap, $l$ a location and $v$ a value by $h[l=v]$ we shall mean the heap that
is the same as $h$ except that it takes on the value $v$ at $l$. By $h(l)$ we
mean the value $e$ to which $h$ maps $l$. We let $H$ denote the space
of all heaps. 

\begin{definition}[Statements]\label{def:grammar}
The statements are defined by the productions: 

\begin{tabular}{lllll}
\multicolumn{2}{l}{(Statements)} & $s$ & {::=} & \\
&\multicolumn{2}{l}{b;} &\multicolumn{2}{l}{\mbox{\em Basic statements.}}\\
&\multicolumn{2}{l}{\Advance;} &\multicolumn{2}{l}{\mbox{\em Clock advance statement.}}\\
&\multicolumn{2}{l}{$\{t \, s\}$} &\multicolumn{2}{l}{\mbox{\em Execute $s$ then $t$.}}\\
&\multicolumn{2}{l}{$\FOR{x}{e1}{e2}{s}$} &\multicolumn{2}{l}{\mbox{\em Execute $s$ for
  x in  $e1\ldots e2$.}}\\
&\multicolumn{2}{l}{\async{s}} &\multicolumn{2}{l}{\mbox{\em Spawn $s$.} }\\
&\multicolumn{2}{l}{\finish{s}} &\multicolumn{2}{l}{\mbox{\em Execute $s$ and wait for termination.}}\\
\end{tabular}

We will assume that the set of basic statements includes \Skip, 
a statement that immediately terminates with an unchanged heap.
\end{definition}
We let $S$ denote the space of all statements.

Procedures (methods) can be defined in X10 in a manner familiar
from object-oriented programming languages like Java. Unlike Cilk,
concurrency constructs in X10 can cross procedure boundaries. For
example, the body $s$ of a $\finish{s}$ or an $\async{s}$ could
contain a call to a method whose body contains an $\async{t}$, a
$\finish{t}$, or an $\Advance$.  For the purposes of this paper we do
not formally model procedure calls, leaving it for future work. The
main implication is that, in general, static concurrency analysis of X10
programs involves inter-procedural reasoning. 

\paragraph{Execution relation.}
As is conventional in SOS, we shall take a {\em configuration} to be a pair
\tuple{s,h} (representing a state in which $s$ has to be executed in the heap
$h$) or $h$ (representing a terminated computation.) Formally, the
space of configurations $K$ is given by $K= (S \times H)+H$, where
$\times$ represents the cross product of two sets and $+$ their
disjoint union.

The operational execution relation $\derives$ is defined as a binary relation
on configurations.  We use the ``matrix'' convention for presenting rules
compactly. A rule such as:
$$
\from{
  \begin{array}{l}
    c\\
    \gamma \derives \gamma_0 \alt \ldots\alt\gamma_{n-1}
  \end{array}
}
\infer{
  \begin{array}{l}
    (d_0)\ \gamma^0 \derives \delta^0_0 \alt \ldots\alt\delta^0_{n-1}\\
    \ldots\\
    (d_{m-1})\ \gamma^{m-1} \derives \delta^{m-1}_0\alt\ldots\alt \delta^{m-1}_{n-1}\\
  \end{array}
    }
$$
\noindent (with $p\geq 0,m>0,n>0$) is taken as shorthand for $m\times n$
rules: infer $\gamma^i \derives \delta^i_j$  from
$c, d_i\ \gamma\derives \gamma_j$, for $i<m,j<n$. Here, $c$  and $(d_i)$ are sequences of conditions, omitted when the sequence is empty.

To define the axioms and rules of inference we need two auxiliary structural predicates on statements. They define what it means for a statement to be {\em asynchronous} and {\em synchronous}. A statement is asynchronous if it is an \async{s}, or a sequential composition of asynchronous statements. 
\begin{equation}\label{prop:async-context}
\begin{array}{cc}
\vdash \isasync{\async{s}} &
\from{
\vdash \isasync{s}
}\infer{
\vdash \isasync{\FOR{x}{e_1}{e_2}{s}}
}\\
\quad\\
\from{
\begin{array}{cc}
\vdash \isasync{s} & \vdash \isasync{t}
\end{array}
}\infer{
\vdash \isasync{\{s,t\}}
}
\end{array}
\end{equation}
A statement is synchronous if it is not asynchronous. We can give a positive definition thus:
\begin{equation}\label{prop:sync-context}
\begin{array}{cc}
\vdash \issync{b;} &
\vdash \issync{\Advance{};} \\
\quad\\
\vdash \issync{\finish{s}} &
\from{
\vdash \issync{s}
}\infer{
\begin{array}{l}
\vdash \issync{\{s\,t\}} \\
\vdash \issync{\{t\,s\}}\\
\vdash \issync{\FOR{x}{e_1}{e_2}{s}} 
\end{array}
}
\end{array}
\end{equation}

The following proposition is established by structural induction.
\begin{proposition}
For any statement $s$, either $\vdash \isasync{s}$ xor $\vdash \issync{s}$.
\end{proposition}
Now we turn to the axiomatization of the transition relation. This is the same as \cite{yuki2013array} except that we permit the first statement of a sequential composition to be arbitrarily nested. 

\paragraph{Basic statements.} \label{sem:basic}
We assume that basic statements come with a definition of the inference relation on configurations. For instance, if assignment were included as a basic statement, it would be formalized thus:
\begin{equation}\label{rule:base}
  \from{l= h(a), v=h(e)}
\infer{\tuple{a=e,h} \derives h[l=v]}
\end{equation}

Similarly, \Skip{} is formalized thus:
\begin{equation}\label{rule:skip}
\tuple{\Skip,h} \derives h
\end{equation}

\paragraph{Sequencing, finish, async.}
\begin{equation}\label{rule:seq}
  \from{\tuple{s, h} \derives \tuple{s', h'} \alt h'}
\infer{
  \begin{array}{rcll}
    \tuple{\async{s}, h} &\derives& \tuple{\async{s'}, h'} &\alt h'\\
    \tuple{\finish{s}, h} &\derives& \tuple{\finish{s'}, h'}&\alt h'\\
    \tuple{\{s\, t\}, h} &\derives& \tuple{\{s'\, t\}, h'} &\alt \tuple{t,h'}\\
 (\vdash \isasync{t})  \tuple{\{t\, s\}, h} &\derives& \tuple{\{t\,s'\}, h'}&\alt \tuple{t,h'}\\
  \end{array}}
\end{equation}
As in \cite{yuki2013array} we can read these rules as specifying which substatements are activated when the statement itself is activated. The last (``out of order'') rule captures the essence of asynchronous execution: it permits the second statement to be activated if the first is asynchronous. 

\paragraph{For loop transitions.}
These rules are unchanged from \cite{yuki2013array}.
The first \texttt{for} rule terminates execution of the \texttt{for} statement if its
lower bound is greater than its upper bound.

\begin{equation}\label{rule:for-base}
  \from{l= h(e_0), u=h(e_1), l>u}
\infer{\tuple{\FOR{x}{e_0}{e_1}{s},h} \derives h}
\end{equation}

The recursive rule performs a ``one step'' unfolding of the \texttt{for} loop.
Note that the binding of $x$ to a value $l$ is represented by applying the
substitution $\theta= x \mapsto l$ to $s$, 
rather than by adding the binding to the heap. This is permissible because $x$
does not represent a mutable location in $s$. 
\begin{equation}\label{rule:for-first}
  \from{
  \begin{array}{rcl}
\multicolumn{3}{l}{l=h(e_0), u=h(e_1), l\leq u, m=l+1,t=s[l/x]}\\
\tuple{t, h} &\derives& \tuple{T', h'} \alt h'
  \end{array}}
\infer{
  \begin{array}{rcl}
\tuple{\FOR{x}{e_0}{e_1}{s},h} &\derives& \\
\multicolumn{3}{l}{\quad\quad\quad\tuple{\{T'\,\FOR{x}{m}{u}{s}\},h'}  \alt     \tuple{\FOR{x}{m}{u}{s},h'}} 
  \end{array}}
\end{equation}

\paragraph{Clock transitions.}
First we define what it means for a statement to be {\em
  stuck}. Intuitively, a statement is stuck if every activated sub-statement is an \Advance. \footnote{We use the matrix notation defined above for transition relations, adapting it for \stuck{s} judgements.}
\begin{equation}\label{def:stuck}
  \begin{array}{c}
    \vdash \stuck{\Advance;} \\ \quad\\
    \from{
      \vdash \stuck{s}
    }\infer{
      \begin{array}{l}
        \vdash \stuck{\async{s}} \\
        (\vdash \issync{s}) \vdash \stuck{\{s\,t\}} \\
        (\vdash \isasync{s}\ \vdash\stuck{t}) \vdash \stuck{\{s\,t\}} 
      \end{array}
    }
  \end{array}
\end{equation}

\begin{proposition}[Stuck configurations are stuck] \label{prop:stuck}
Let $\derives'$ be a transition relation on
  configurations defined by the rules introduced so far, i.e. 
  Rules~(\ref{rule:base},\ref{rule:skip},\ref{rule:seq},\ref{rule:for-base},\ref{rule:for-first}). Then 
  $\vdash   \stuck{s}$ iff for no $t,h,h'$ is it the case that
  $\tuple{s,h} \derives' \tuple{t,h'}$ or $\tuple{s,h} \derives'h'$.
\end{proposition}
\def\derivesNext{\Longrightarrow}
Next we define a relation on statements. We write $s \derivesNext t$ (and say ``$s$ yields $t$ after a clock step''). 

\begin{equation}\label{trans:derivesNext}
  \begin{array}{cc}
     \Advance; \derivesNext \Skip; &
     b; \derivesNext b; \\
\quad\\
    \from{
      s \derivesNext t
    }\infer{
      \begin{array}{rcl}
        \async{s} &\derivesNext& \async{t} \\
        (\vdash \issync{s}) \{s\, u\} &\derivesNext& \{t\,u\}\\
      \end{array}
    } & \\
\quad\\
    \from{
      \begin{array}{ccc}
        \vdash \isasync{s} & s \derivesNext s' & t \derivesNext t'
    \end{array}}
    \infer{
      \{s\,t\} \derivesNext \{s'\,t'\}
    }
  \end{array}
\end{equation}

Now we can define the transition relation. If $s$ is stuck, and $s$ yields $t$
after a clock step, then $\finish{s}$ can transition to $\finish{t}$, and leave
the heap unchanged:
\begin{equation}\label{rule:clock}
 \from{
\begin{array}{cc}
  \vdash \stuck{s}  & s \derivesNext t
\end{array}
}\infer{
 \tuple{\finish{s}, h} \derives \tuple{\finish{t},h}
}
\end{equation}
The Rules~(\ref{rule:base}, \ref{rule:skip}, \ref{rule:seq},
\ref{rule:for-base}, \ref{rule:for-first} and \ref{rule:clock}) complete the
definition of the transition relation $\derives$.

\begin{proposition} The only terminal configurations for the
  transition system $(K, \derives)$ are of the form $h\in H$.
\end{proposition}
\paragraph{Semantics.}
We now define appropriate semantical notions.

\begin{definition}[Semantics]
Let $\starderives$ represent the reflexive, transitive closure of $\derives$. 
The operational semantics, $\O{s}$ of a statement $s$ is the relation
$$\O{s}\defeq \{(h,h') \alt \tuple{s,h}\starderives h'\}$$

Sometimes a set of \emph{observable} variables is defined by the programmer,
and the notion of semantics appropriately refined:
\[\O{s,V}\defeq \{(h,\restr{h'}{V}) \alt \tuple{s,h}\starderives h'\}\]
\noindent where for a function $f:D \rightarrow R$ and $V \subseteq D$ by
$\restr{f}{V}$ we mean the function $f$ restricted to the domain $V$.

\end{definition}

%% file: clocked-H-B.tex
The semantics of Section \ref{sec:sem} is operational: it can be considered
as a blueprint for a rudimentary X10 interpreter. Concurrency is represented
as non determinism. An interpretation of a program $s$ is a linear succession
of reductions according to the rules of Section \ref{sec:sem}. Each reduction
is associated to the execution of a basic statement or to the crossing
of a barrier. Due to the fact that a sequence $\{s \, t\}$ can be reduced
in two ways if $\isasync{s}$ is true, a program can have many interpretations,
all of which are interleaves of the same set of basic operations.

In contrast, in building the happens-before realation, our aim is to
specify the X10 semantics as a partial order: operation $u$ happens
before $v$ if it occurs to the left of $v$ in every legal interpretation
To simplify the presentation, we will assume that basic statements have
distinct names. For polyhedral X10 programs, we can use \emph{paths}
as defined in Section \ref{sect:paths} for that purpose.

Polyhedral X10 programs have static
control: the set of operations and their execution order are fixed as soon
as a few size parameters are known. Hence, we can simplify the semantics
of Section \ref{sec:sem} by dispensing with the heap. A convenient way of
achieving this simplification is to assume that all basic statements
are skips. Under this assumption, all configutations have the same 
heap as the initial configuration of the program.

As a first effort, we will ignore loops in what follows. Note that when
size parameters are given, loops can be statically expanded into nests
of sequence constructs.

Each step in the interpretation of a program is a deduction according to
the rules of inference of Section \ref{sec:sem};
each such deduction must start with one
of the axioms (\ref{rule:skip}) or (\ref{rule:clock}).
Each reduction is associated either to a basic statement or to the set of
advances which are transformed into skips by rule (\ref{rule:clock}).
No reduction can use more than one axiom, since there is no transitivity
rule for $\derives$. To such an interpretation we can associate the trace 
obtained by succesively appending either the name of the reduced skip or 
the set of transformed advances. Observe that at each step, either the 
number of advances or the
number of basic statements is reduced, hence, all reductions terminates. 
This is an indirect proof that the fragment of X10 we consider has no deadlock. 

\begin{proposition}[Normal Forms] Let $r$ be a program which is not stuck.
Either $r$ does not contain advances, or there is a unique stuck $r'$ such
that every reduction path from $r$ 
which does not use (\ref{rule:clock}) terminates in $r'$
\end{proposition}
For reasons to be made clear presently, $r'$ will be called the 
\emph{normal form} of $r$, $\mathcal{N}(r)$. If $r$ does not contain advances,
$\mathcal{N}(r)$ is a terminal configuration: a configuration
without a continuation..

\begin{proof} It is clear that the only case in which several reductions
are possible is if $r = \{s \, t\}$ and $\vdash \isasync{s}$. One can either
reduce $s$, giving $\{s' \, t\}$, or $t$, giving $\{s \, t'\}$. This last
term can be reduced further, giving $\{s' \, t'\}$. It easy to convince
oneself that $s'$ is either terminal or $\vdash \isasync{s'}$. In the 
first case, both terms reduce to $t'$. In the second case, $\{s' \, t\}$
reduces to  $\{s' \, t'\}$. In both cases, the reduction system has the
\emph{weak diamond property}, and since reductions terminate, the
unique normal formal form or Church-Rosser \cite[Chapter~3]{Bare:84} 
property, Q.E.D.
\end{proof}

This result is the key to the analysis of clocks in polyhedral X10
programs.

\subsection{The Case of a Single Clock Program}

Consider first a one clock program, i.e. a program of the form 
$\finish{r}$ where $r$ does not contain any clocked finish. This program
has a unique normal form $\finish{s}$ where $s$ is stuck. The reduction can 
only progress by applying rule (\ref{rule:clock}), giving a new program
$\finish{t}$ which may be further reduced. The elaboration of the
initial program therefore proceeds in \emph{phases}, where only rules
(\ref{rule:skip}-\ref{rule:seq}) are applied, separated by applications
of (\ref{rule:clock}). One convenient way of expressing these observations
is to assign a number in sequence to each application of 
rule (\ref{rule:clock}), and to assign a phase number
$\phi(u)$ to each operation $u$ which is executed between applications
$\phi(u)$ and $\phi(u)+1$ of (\ref{rule:clock}). It is then obvious that 
if $\phi(u) < \phi(v)$ 
then $u \hbk v$. The fragment of code which constitutes a phase fits
into the model of \cite{yuki2013array}, and hence has the same happens-before
relation. The clocked happens-before relation is therefore:
\[ u \hbk v \equiv \phi(u) < \phi(v) \vee u \prec v .\]

\subsection{Multiple Clocks}\label{sec:multi-clockHB}

Let us now consider two operation $u$ and $v$ in a multiple clock program.
There exists an innermost finish $F$ which contains both $u$ and $v$.
In $F$, there exists an outermost finish $f_u$ (resp. $f_v$) which
contains $u$ (resp. $v$). Either $f_u$ or $f_v$ or both may be the same
as $F$. If $f_u = f_v = F$, we are back to a single clock program
and the conclusion of the preceding section stands.

Suppose now that neither $f_u$ nor $f_v$ are equal to $F$. In the text
of the program, replace $f_u$ (resp. $f_v$) by a fictitious basic 
statement $U$ (resp. $V$) and evaluate $U \hbk V$, again by the method
of the preceding section. We claim that $u \hbk v \equiv U \hbk V$.
Assume first that $U \hbk V$ is true. In all reductions of the transformed
$F$, $U$ is reduced before $V$. We can construct a reduction of the
original $F$ by replacing the reduction of $U$ by the reduction of
$f_u$, and, by the semantics of finish, no operation of $f_v$ will
execute until $f_u$ has terminated. Hence, $u \hbk v$ is true. The case
$V \hbk U$ is symetric. If neither $U \hbk V$ nor $V \hbk U$ are true,
then there is a reduction in which $U$ and hence $u$ occurs first, and another 
one in which $V$ hence $v$ occurs first, hence neither $u \hbk v$ nor
$v \hbk u$ are true. The remaining two cases can be handled in the same way.

\subsection{Loops}\label{sec:HBloops}
The analysis of loops poses both a practical and a theoretical problem.
On the practical side, when loop bounds are known numbers, loops can be
eliminated in favor of sequences by repeated application of rules
(\ref{rule:for-first}-\ref{rule:for-base}). However, the resulting program
may be so large as to make counting phases unpractical. But in the polytope
model, loop bounds may depend on unknown symbolic parameters, hence the
application of rule (\ref{rule:for-first}) may never terminate. The trick
here is to observe that a program which depends on symbolic constants is
a shorthand for a possibly infinite family of programs, which are obtained by
giving every admissible value to the parameters. For each program in
the family, the conclusions of the preceding section still stand. All that
is needed is to find a closed form for the advance counters $\phi$ and
for the unclocked happens-before relation, $\prec$. For the later, the
authors of \cite{yuki2013array} have given a closed formula as an incomplete
lexicographic order. It remains to find a symbolic way of computing $\phi$.

To this aim, assume that advances are temporarilly considered as ordinary
basic statements, to which $\prec$ applies. One way of interpreting the fact
that a configuration is stuck is to say that no reduction can be done
unless one advance at least is reduced, which can be done only by using
rule (\ref{rule:clock}). The initial advances of a stuck $s$ are the
advances which are replaced by skips before application of (\ref{rule:clock}).

\begin{proposition} If $s$ is stuck, then for every elementary
statement $x$ in $s$, either $s = x$ is an initial advance, or
there exists a unique initial advance $a$ such that $a \prec x$.
\end{proposition}

\begin{proof} The proof is by induction on the $\vdash \stuck{s}$ inference. 
There is nothing to prove if $s$ is an advance. If $s$ is $\async{t}$, then
$t$ is stuck and $x$ is in $t$, and the result follows. If $s$ is
$\{t \, u \}$ and $t$ is asynchronous, then both $t$ and $u$ are stuck
and $x$ is in either $t$ or $u$. If $t$ is synchronous, then $t$ is stuck.
If $x$ is in $t$, then the induction hypothesis applies. If $x$ is
in $u$, let $a$ be the initial adavance of $t$, which exists by the
induction hypothesis. By the semantics of sequential composition, $a \prec x$,
Q.E.D.
\end{proof}

Let $\mathcal{A}$ be the set of advance instances in the program under study.
From the above result follows that, for each statement $u$:

\[ \phi(x) = \mathrm{Card} \{ a \in \mathcal{A} | a \prec x\} .\]

For polyhedral programs, we will show in Section \ref{barvinok} how to
compute closed forms for $\phi$.

%% file: races.tex
The important observation is that clocks only add additional synchronizations
among activities, and hence strictly decreases the set of
may-happen-in-parallel iterations.  Therefore, we propose to guarantee
race-freedom by disproving all races found with out taking clocks into
account~\cite{yuki2013array}.

\subsection{Computing the $\phi$ Function}\label{barvinok}
The first step is to automatically compute the $\phi$ functions that define the
happens-before relation with respect to clocks.  What we are interested in is a
function that gives the number of advances, associated with a clock, an
activity has executed before executing an iteration of a statement. Thus, the
computed function must be parametric to the statement \emph{instance} in
question, as well as the program parameters.

This is achieved by constructing the following union of
parametric polytopes, and computing the number of integer points it contains:
\begin{eqnarray}
   x & \in  & D_S \label{phi:statement}, \\
   a & \in  & D_A, \label{phi:advance}\\ 
   a & \prec & x \label{phi:hb}
\end{eqnarray}
where 
\begin{itemize}
\item $D_S$ is the domain of a statement $S$,
\item $D_A$ is the union of domains of {\tt advance} statements that operates with the clock in question, and
\item Constraint (\ref{phi:hb}) restricts the {\tt advance} statement instances to those that happens-before an instance $x$ of $S$.
\end{itemize}
Note that this corresponds to the definition of the $\phi$ function in
Section~\ref{sec:HBloops}.  By treating $x$ as parameters of the polytope, we
obtain a parametric expression of the number of {\tt advance} statements that
happens-before $x$.  We compute such expression for each pair of statement and
clock in the program.

\subsection{Counting Integer Points in Polytopes}
For polyhedral iteration spaces, the question
of counting advances can be cast as counting the number of integer points in
polyhedra.  Ehrhart~\cite{ehrhart1962polyedres} showed that the number of
integer points in a polytope can be expressed as \emph{periodic polynomials}.
We use an algorithm proposed by Verdoolaege et
al.~\cite{verdoolaege2007counting} for computing such polynomials, which
handles parametric polytopes.

\subsection{Disproving Races}
We may refine the dataflow analysis formulation for X10 programs without clocks
(overviewed in Section~\ref{sec:ppopp-ada}) using the new happens-before
relation for clocked programs. The only change required is to replace $\prec$ with $\hbk$.
However, the problem stems from the $\phi$ functions not being affine in general.
Parametric integer linear programming~\cite{Feau:88b} can no longer be used,
and there is no known alternative that can handle polynomial expressions.

Therefore, our proposed solution is to first detect races without clocks taken
into consideration, and then later use constraint solvers to verify if the statement
instances involved in a race can take the same value of $\phi$.

\subsection{Problem Formulation}
The races detected have precise information regarding which statement instances are involved in a race.
Recall that we have two main kinds of races (Section~\ref{sec:ppopp-ada}), Read-Write and Write-Write.
For each detected race, we have the following:
\begin{itemize}
\item $r \in D_R^*$: Read instances involved in the race.
\item $w \in D_W^*(r)$: Write instances involved in the race, \emph{parametric} to $r$.
For Write-Write races, we obtain two of such sets, as we have two writers in conflict.
\end{itemize}

Given a set of clocks $C$ in the program, and $\phi$ functions $\phi_c$ for
each $c \in C$.  Recall (Section~\ref{sec:multi-clockHB})
that among the set of clocks, only one clock is relevant for each pair of statements,
and thus the case with multiple clocks was reduced to single clock case.
Let us define ${\tt reduce}(C, S1, S2)$ to be a function that gives a single clock, $c^*$, 
from the set of clocks $C$ that is relevant when defining the happens-before relation between
instances of statements $S1$ and $S2$.

The problem is to simply check if there exists a pair of instances that are involved in a race,
and take the same value of $\phi_{c^*}$ for each potential race.


The constraints for Read-Write race to occur, and respectively for Write-Write race to occur, 
are the following:

\begin{minipage}{0.425\columnwidth}
\begin{align*}
r &\in D_R^*\\
w &\in D_W^*(r)\\
c^* &= \mathtt{reduce}(C, R, W)\\
\phi_{c^*}(r) &= \phi_{c^*}(w)
\end{align*}
\end{minipage}
\vline
\hspace{0.025\columnwidth}
\begin{minipage}{0.45\columnwidth}
\begin{align*}
w1 &\in D_{W1}^*(r)\\
w2 &\in D_{W2}^*(r)\\
c^* &= \mathtt{reduce}(C, W1, W2)\\
\phi_{c^*}(w1) &= \phi_{c^*}(w2)
\end{align*}
\end{minipage}

\subsection{Undecidability}
To prove that the race problem for clocked X10 is undecidable, we need the
following construction:

Given an arbitrary polynomial $P(x)$ in $n$ variables $x_1,\ldots,x_n$ with
integer coefficients, build an X10 program that has a race if and only
if $P(x)$ has an integral root.

Since deciding if $P(x)$ has a root is undecidable (Hilbert 10th problem),
it follows that the race problem is undecidable.

As we will see later, we may have to build not one X10 program but
a finite number of programs (in fact, $2^n$ programs) such that $P(x)$
has a root iff one of those programs has a race, but this does not
change the conclusion.

\subsubsection{The Shape of the Test Program}
Let us write $P(x) = P_1(x) - P_2(x)$, and consider the following:

\begin{verbatim}
  for(x in D){
    clocked finish{
      clocked async{
        L1;           
        u = f();     //U
      }
      clocked async{
        L2;
        g(u);        //G
      }
    }
  }
\end{verbatim}
where $x$ is in fact a vector of dimension $n$ which scans the domain $D$
to be defined later, and $L1$ (resp. $L2$) is a loop nest which executes
exactly $P_1(x)$ (resp. $P_2(x)$) advances. It is clear that this program
will have an MHP race iff $P(x)$ has a root in $D$. 
\begin{quote}\small
To understand the behaviour of this program, assume first, without
loss of generality, that for a given value of $x$, $P_1(x) < P_2(x)$. 
Then the loops $L1$ and $L2$ will execute in lockstep, $L1$ will 
terminates first and $U$ will be executed. The first
activity will terminate and de-register itself, and then $L2$ will 
execute its remaining iterations, and \emph{then} execute $G$.
Contrarywise, if $P_1(x) = P_2(x)$, $L1$ and $L2$ will
terminate at the same (logical) time, and the execution order of 
$U$ and $G$ will be undefined.
\end{quote}

However,
we must first insure that the program is realistic. Observe that the
number of iterations of a loop nest can never be negative; hence, we
must insure that $P_1(x)$ and $P_2(x)$ are non-negative for $x \in D$.
This can be guaranteed if $D$ is the positive orthant and if
$P_1$ and $P_2$ have positive coefficients. In this way, we will test
only the existence of a \emph{positive} integral root of $P$. To be
complete, we must apply the above construction to  the $2^n$ polynomials 
$P(\epsilon_1 x_1, \ldots, \epsilon_n x_n)$, where the $\epsilon$s
take all combinations of the values $\{+1, -1 \}$.

The fact that $D$ is unbounded is irrelevant, since we do not intend to
run the above program to find races. We just have to be careful in 
writing the $x$ loop, but this is a well know problem (see for instance
the classical proof that $\mathrm{Card}\; N = \mathrm{Card}\; N^2$).

\subsubsection{Constructing Counting Loop Nests}
Our aim now is, given a polynomial $Q(x)$ with positive integer
coefficients, to construct a loop nest which compute $Q(x)$ only using
increments. It will be enough then to replace each incrementation
by an advance to prove the theorem. In the following, we accept
more general forms of increments \texttt{phi~+=~d;} where $d$ is a positive
integer, representing $d$ consecutive advances.

Let us select one variable, say $x_1$ and let us write $Q(x) = Q(x_1, x_r)$
where the vector of variables $x_r$ may be empty. Let $m$ be the degree of $Q$ 
in $x_1$. The first difference of $Q$ is:
\begin{equation}
   Q^{(1)}(x_1, x_r) =  Q(x_1+1, x_r) - Q(x_1,x_r)
\end{equation}
and it is clear that the program

\begin{verbatim}
   phi = Q(0, x_r);
   for(i=0; i<x; i++)
      phi += Q1(i, x_r);
\end{verbatim}
compute $Q(x)$. The degree of $Q^{(1)}(i, x_r)$ is $m-1$ in $i$, 
hence we can iterate this construction to obtain:
\begin{verbatim}
   phi = Q(0, x_r);
   for(i=0; i<x; i++) {
/*    phi += Q1(i, x_r); */
      phi += Q1(0, x_r);
      for(j=0; j<i; j++) {
         phi += Q2(j, x_r);
      }
   }
\end{verbatim}
where $Q^{(2)}$ is the second difference of $Q$  with respect to $x_1$.

We can continue in this way until we reach the $m$-th difference, which is 
independent of its first variable. At this point, all the increments in
the program depend only on $x_r$. We can select another variable
and apply the same construction, until all increments are constant. This
terminates the proof.

\subsubsection{An Example} Let us construct the counting nest for 
$Q(x,y) = x^2 + xy + y^2$. The first difference is $2x+y+1$, hence the
first program is:
\begin{verbatim}
  phi = y*y;
  for(i1 = 0; i1<x; i++)
    phi += 2*i1+y+1;
\end{verbatim}
The second difference is simply 2, hence the second program:
\begin{verbatim}
  phi = y*y;
  for(i1 = 0; i1<x; i1++){
    phi += y+1
    for(i2 = 0; i2<i1; i2++)
      phi+= 2;
  }
\end{verbatim}
At this point, the increments depend only on $y$. Applying 
the same algorithm to $y^2$ and $y+1$ yields the final program:

\begin{verbatim}
  phi = 0;
  for(i3 = 0; i3<y; i3++){
    phi += 1;
    for(i4 = 0; i4<i3; i4++)
      phi += 2;
  }
  for(i1 = 0; i1<x; i1++){
    phi += 1;
    for(i5 = 0; i5<y; i5++)
      phi += 1;
    for(i2 = 0; i2<i1; i2++)
      phi += 2;
  }
\end{verbatim}
Note that a new induction variable is needed for each loop.

%% file: examples.tex

In this section, we illustrate different types of synchronizations that can be
analyzed by our proposed approach through examples.

\subsection{Barrier-like Synchronization}
The following is a simplified implementation of 1D Jacobi-style stencil
computation using clocks.

\begin{lstlisting}
clocked finish
   for (i=1:N-1)
     clocked async
        for (t=0:T) 
           B[i] = S0(A[i-1], A[i], A[i+1]);
           advance;
           A[i] = S1(B[i-1], B[i], B[i+1]);
           advance;
\end{lstlisting}

The above use of clocks are similar to barriers;
the synchronization does not rely on an activity being \emph{de-registered}
from a clock, aside from the primary activity.

Without taking the additional happens-before relation due to clocks into
account, there are four read-write races in the program.

\begin{itemize}
\item Read $A[i-1]$ by $S0\langle i,t\rangle$ is in read-write race with $S1\langle i-1,t'\rangle$ when $1<i<N$ and $0\le t'\le T$.
\item Read $A[i+1]$ by $S0\langle i,t\rangle$ is in read-write race with $S1\langle i+1,t'\rangle$ when $1\le i<N-1$ and $0\le t'\le T$.
\item Read $B[i-1]$ by $S1\langle i,t\rangle$ is in read-write race with $S0\langle i-1,t'\rangle$ when $1<i<N$ and $0\le t'\le T$.
\item Read $B[i+1]$ by $S1\langle i,t\rangle$ is in read-write race with $S0\langle i+1,t'\rangle$ when $1\le i<N-1$ and $0\le t'\le T$.
\end{itemize}
Note that $t'$ refers to all possible values that $t$ can take. Without the clock synchronization,
all writes to the same element of the array at different time steps are in conflict.

The $\phi$ functions for $S0$ and $S1$ are $\phi_{S0} = 2t$, and $\phi_{S1} = 2t+1$.
The four races can trivially be disproved by using the $\phi$ functions, since
it guarantees $S0$ and $S1$ never execute in parallel.  In this case, the
$\phi$ functions are actually affine, and hence it can directly be incorporated
into array dataflow analysis.

\subsection{Activity Specific $\phi$ Functions}
The following is a parallelization of Gauss-Seidel style stencil computation.
The difference is that the reference $A[i-1]$ refers to a value computed at the
same $t$, rather than $t-1$ in the case of Jacobi style stencils.

\begin{lstlisting}
clocked finish
   for (i=1:N-1)
     clocked async
        for (t=0:T)
           advance;
           A[i] = S0(A[i-1], A[i], A[i+1]);
     	   advance;
     advance;
\end{lstlisting}

The parallelization using clocks is illustrated in Figure~\ref{fig:example-GS}.

\begin{figure}
        \begin{subfigure}[b]{0.5\columnwidth}
                \centering
                \includegraphics[width=0.85\columnwidth]{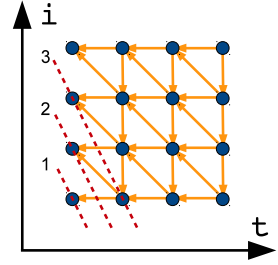}
		\caption{\label{fig:example-GS-original}Original iteration space, and its dependences.
 The dashed lines show sets of parallel iterations.}
        \end{subfigure}
	\vspace{-2.8cm}
	\\
        \begin{subfigure}[b]{0.5\columnwidth}
                \centering
                \includegraphics[width=0.85\columnwidth]{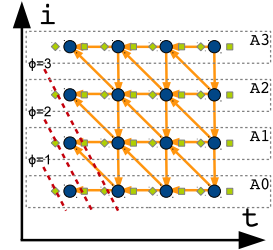}
		\caption{\label{fig:example-GS-X10} Parallelization with X10
using clocks. An activiy execute one row of iterations with {\tt advance} statements
before and after the computation. First three synchronizations are shown as dashed lines.}
        \end{subfigure}
	\hspace{0.05\columnwidth}
	\begin{subfigure}[b]{0.45\columnwidth}
                \centering
                \includegraphics[width=0.9\columnwidth]{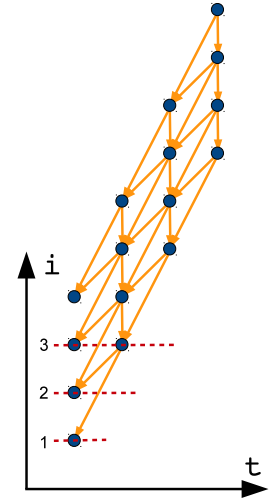}
		\caption{\label{fig:example-GS-skewed}{\tt doall} parallelism
with skewing. The iteration space may be skewed to align the parallel
iterations to one of the axes to enable {\tt doall} type parallelism.}
        \end{subfigure}
	
	\caption{\label{fig:example-GS} Parallelization of Gauss-Seidel Stencil with clocks. }
\end{figure}

There are two read-write races in the program before clocks are taken into consideration.
\begin{itemize}
\item Read $A[i-1]$ by $S0\langle i,t\rangle$ is in read-write race with $S0\langle i-1,t'\rangle$ when $1<i<N$ and $0\le t'\le T$.
\item Read $A[i+1]$ by $S0\langle i,t\rangle$ is in read-write race with $S0\langle i+1,t'\rangle$ when $1\le i<N-1$ and $0\le t'\le T$.
\end{itemize}
The $\phi$ function for $S0$ is $\phi_{S0} = 2t+i$.

Using the $\phi$ function, constraint solvers can easily disprove these races.
The problem reduces to the existence of values of $t$ and $t'$ within its domain that satisfies:
\begin{itemize}
\item $2t+i = 2t'+i-1$, or
\item $2t+i = 2t'+i+1$. 
\end{itemize}
It can easily be found that the LHS is limited to even numbers and the RHS is limited to odd numbers.

Note that the $\phi$ function found involves $i$ that takes different values in
each activity.  This is because the primary thread also executes an {\tt
advance} statement after spawning each thread, allowing earlier activities to
proceed. The number of activities participating in a synchronization
dynamically changes as the program executes, which is different from how
barrier synchronization typically work.

\subsection{Polynomial $\phi$ Functions}
The following is a possible parallelization of QR decomposition (via CORDIC) using clocks.
The advance statement is surrounded by two loops and an affine {\tt if} guard.
The statement domains are no longer rectangular like other examples, leading to more complicated $\phi$ functions.

\begin{lstlisting}
clocked finish
  for (j=0:N-1)
    clocked async
      for (k=0:N-2)
        for (i=0:N-2-k) 
          if (j>=k) {
            M[N-i-1][j] = 
                S0(M[N-i-1][j], M[N-i-2][j],
                   M[N-i-1][k], M[N-i-2][k]);
            M[N-i-2][j] = 
                S1(M[N-i-1][j], M[N-i-2][j], 
                   M[N-i-1][k], M[N-i-2][k]);
            advance;
          }
\end{lstlisting}

There are a total of eight Read-Write races in the above program. We do not enumerate
them as they are quite similar. Writes to $M$ by both statements are in race
with reads where $k$ is used in the access. 

The $\phi$ functions for two statements are the following:
\begin{itemize}
\item $\phi_{S0} = \phi_{S1} = Nk + i - \frac{k^2+k}{2}$
\end{itemize}

We illustrate the problem formulation with one of the eight races.
Given a writer instance $S0\langle j,k,i\rangle$, and a read access $M[N-i-1][k]$ of instance $S1\langle j',k',i'\rangle$, there
is a race when $j>k$ and $i=i$' and $k=j'$ and $k >= k'$ and $i<=N-2-k'$.

To disprove this race, we must ensure that there is no pair of distinct instances that satisfies:
\begin{itemize}
\item $i=i'$, $k=j'$, $k>=k'$, $j>k$, and
\item $Nk+i-\frac{k^2+k}{2} = Nk'+i'-\frac{k'^2+k'}{2}$.
\end{itemize}
With constraint solvers that can handle polynomials over integers, it can be verified that the above cannot be satisfied.


%% file: implementation.tex
We have implemented\footnote{Our implementation is open source, and will be made available when the
paper is published.} our analysis for the subset of X10 we handle. We take
a simplified representation of the program only concerning the access to
variables, disregarding the specifics of the operations performed in statements.
This simplified representation can easily be extracted from the full X10 AST.

Significant amount of effort has been made towards identifying polyhedral
regions of loop programs~\cite{grosser2011polly, pop2006graphite}.  
However, integrating such effort to X10 is beyond the scope of this paper.

We use the Integer Set Library~\cite{verdoolaege2010isl} for polyhedral operations and parametric integer programming,
and the Barvinok library~\cite{verdoolaege2007counting} for counting of integer points.
Other parts of the analysis are implemented in Java, and we use native bindings for library calls.

We use the Z3 SMT solver~\cite{moura2008z3} for disproving races involving polynomials.
The constraints are given to Z3 using the SMT-LIB format~\cite{barrett2010smtlib2}, 
and many other solvers that support the same standard can also be used.
We interface with Z3 through the command line.

\subsection{Z3}
We require non-linear arithmetic over integers, which is
supported only by a small subset of available constraint solvers. Z3 is a
robust and efficient SMT solver that has these features.

As our problem is undecidable in general, we have explored a number of
simplifications we may perform to help the constraint solvers. For example, when two
instances of a same statement are involved in a potential race, the $\phi$
functions are identical. The $\phi$ functions may be identical even when multiple
statements are involved, depending on the placement of {\tt advance}
statements. 

Many such properties about $\phi$ functions may be deduced by analyzing loop
structures.  Taking advantage of these can significantly simplify the problem
formulation. However, we found the Z3 implementation to be quite efficient for
programs that we tested, and hence these simplifications were not implemented.

\subsection{Clocked Benchmarks} \label{sec:benchmarks}
As we have illustrated in Section~\ref{sec:examples}, our main contribution is
in expanding the class of X10 programs that can be analyzed.  In this section,
we use a subset of benchmarks from Java Grande Forum benchmark
suite~\cite{smith2001parallel} parallelized using clocks, and the examples from
Section~\ref{sec:examples} to show the performance of our analysis. 

The benchmarks from JGF that we use are {\sc SOR}, {\sc LUFact}, and {\sc
MolDyn}.  Since the parallelization of {\sc SOR} and {\sc LUFact} in X10
benchmark repository do not use clocks, we took the parallelization using
\emph{phasers}; a synchronization construct closely related to clocks in
Habanero-Java~\cite{cave2011habanero}.  Although formal handling of HJ is not
the scope of our paper, our analysis on X10 can be carried over to its dialect,
Habanero-Java. These benchmarks in X10 and HJ were manually converted to the input
expected by our tool. 

Table~\ref{tab:JGF} summarizes our result. With the exception of {\sc MolDyn}, which has
34 statements and 9 advances, the time spent on counting advances is relatively small.
For all cases, Z3 was able to disprove the race quickly.

\begin{table}
\begin{threeparttable}[b]
\begin{tabular}{|c|c|c|c|c|c|}
\hline
{\bf Benchmark} & Races & DFA      & Counting  & SMT      & Total\\
                & Found & Time & Time & Time & Time \\
\hline
{\sc SOR}  & 4 & 3.02 & 0.54 & 0.18 & 3.57\\
{\sc MolDyn}\tnote{1}  & 16 & 5.35 & 6.55 & 0.33 & 11.95\\
{\sc LUFact}\tnote{2}  & 5 & 0.71 & 0.25 & 0.15 & 0.97\\
{\sc Jacobi}  & 4 & 0.32 & 0.23 & 0.09 & 0.56\\
{\sc Gauss-Seidel}  & 2 & 0.19 & 0.25 & 0.07 & 0.44\\
{\sc QR} & 8 & 1.35 & 0.22 & 0.29 & 1.58\\
\hline
\end{tabular}
\caption{\label{tab:JGF} Performance (in seconds) of our implementation on JGF
benchmarks and examples in Section~\ref{sec:examples}.  All races found were
disproved by the SMT solver. Our experiments were conducted with 4 core Intel
i7 (2.4 GHz) and 4GB of memory. We used Java 1.7, ISL 0.11, Barvinok 0.36, and Z3 4.3.1.}
\begin{tablenotes}
\item [1] {\sc MolDyn} uses a class to represent particles. We treat array of
particles as data arrays, and any change of the object members are considered
as writes to the array for the purpose of our analysis.
It also contains a data-dependent {\tt if} statement, which decides if a particle is updated or not.
We handle such case by conservatively assuming that the branch is always taken (the particle is always updated.)
\item [2] {\sc LUFact} in JGF does pivoting making it non-polyhedral. We have
removed the pivoting to apply our analysis.
\vspace{-0.08cm}
\end{tablenotes}
\end{threeparttable}
\end{table}

%% file: related.tex
In this section, we place our work in context of the previous work on analysis
of barrier synchronization.

\subsection{Analysis of Clocks}
There are very little work in the literature that deals with clocks in X10 (or
similar parallel constructs).  Vasudevan et al.~\cite{vasudevan2009compile}
have presented static analysis for verifying \emph{syntactic} correctness of
clocks to avoid runtime exceptions.  They also propose a few optimizations for
certain patterns of clock usage, e.g., when some activities are registered but
do not participate in synchronization.

We distinguish our work by providing race free guarantee to the polyhedral
subset of X10 programs.  We are not aware of any other work that verify absence
of races in the presence of clocks.

\subsection{Analysis of Barriers}
Barrier synchronization in conventional parallel programming models (e.g., MPI,
OpenMP) has similar semantics as X10 clocks. There are tools for dynamic
analysis of these languages~\cite{vetter2000umpire, krammer2004marmot,
park2011efficient}. 
Our work complements these tools by providing race-free guarantee for program
regions amenable to our static analysis.

There are static analysis techniques developed for analyzing barrier
synchronization of SPMD-style programs~\cite{aiken1998barrier,
kamil2005concurrency, zhang2007barrier}. SPMD (Single Program Multiple Data) is
a common parallel programming model where the same code is executed on multiple
processes. Typically, each process work on different data by referring to its
process ID and synchronize/communicate accordingly.

Aiken and Gay~\cite{aiken1998barrier} introduced
\emph{single-valued} expressions, expressions that can be proved to evaluate to
the same value in all processes, to ensure that all processes execute the same
number of barriers.  Kamil and Yelick~\cite{kamil2005concurrency} extend the
work of Aiken and Gay in the context of Titanium parallel programming
language~\cite{yelick1998titanium} to perform May-Happen-in-Parallel analysis.

Titanium requires all the barriers to be \emph{textually aligned}, i.e., all
processes must execute the same barrier in the same order.  Zhang and
Duesterwald~\cite{zhang2007barrier} present a method for more general
SPMD-style programs where barriers are not necessarily textually aligned.

X10 is not SPMD; activities are dynamically created executing its own piece of
code.  Furthermore, the participants of barriers can dynamically change,
complicating static analysis.

\subsection{Handling of {\tt at} and {\tt places}}
X10 uses Partitioned Global Address Space (PGAS) programming model, where 
the address space is spearated into multiple \emph{places}. The {\tt at}
construct allows the programmer to specify the place where operations
are  performed.

Agarwal et al. presented an algorithm to find the set of iterations that may
happen in parallel for X10 programs~\cite{agarwal2007may}.  Although they do
not handle {\tt clocks}, their algorithm handles {\tt at} and \emph{places}.
They assume that places are identified as some function of the loop indices.
Two statement instances execute at the same place only if the expressions
evaluate to the same value.

They also handle {\tt atomic} blocks in X10, which are similar to {\tt atomic}
blocks in other languages, but only allow concurrent execution of critical
sections if two processes are in different places.

Both of these may be used in combination with our work to further extend the
applicability of our analysis.

%% file: conclusion.tex
We have presented a method for guaranteeing race freedom of polyhedral X10
programs with clocks. We show that the problem is undecidable and resort to
constraint solvers for providing the guarantee. The idea is not limited to X10
programs, and can easily be adapted to handle its dialect, Habanero-Java, and
possibly other languages with less dynamic synchronizations.  When combined
with the work by Agarwal et al.~\cite{agarwal2007may}, we may now analyze all
of the basic parallel constructs in X10. 

There has been little work on static analysis of parallel programming
languages.  The application of the main ideas in our work is not limited to
data race detection.  The formalization and happens-before relation opens many
opportunities such as scheduling, memory allocation, program transformations,
and so on.

Another direction of future work is to relax the polyhedral requirement.
This is a limitation shared among any polyhedral analysis, and our work will also
benefit from progress in this direction.
